\documentclass{article} 
\usepackage{iclr2021_conference,times}
\usepackage{tikz}
\usepackage{amsmath}
\usepackage{amsthm}
\usepackage{amssymb}
\usepackage{comment}
\usepackage{algorithmic}
\usepackage{algorithm}
\usepackage{enumitem}
\usepackage{xcolor}
\usepackage{wrapfig}
\usepackage{caption}
\usepackage{subcaption}
\usetikzlibrary{fit,calc}
\colorlet{pink}{red!40}
\colorlet{blue}{cyan!60}
\colorlet{yellow}{green!60}

\usepackage{amsmath,amsfonts,bm}

\def\eqref#1{equation~\ref{#1}}

\def\1{\bm{1}}

\DeclareMathAlphabet{\mathsfit}{\encodingdefault}{\sfdefault}{m}{sl}
\SetMathAlphabet{\mathsfit}{bold}{\encodingdefault}{\sfdefault}{bx}{n}

\usepackage{hyperref}
\usepackage{url}

\newtheorem{thm}{Theorem}

\newtheorem{cor}{Corollary}

\newcommand*{\tikzmk}[1]{\tikz[remember picture,overlay,] \node (#1) {};\ignorespaces}
\newcommand{\boxit}[1]{\tikz[remember picture,overlay]{\node[yshift=0pt,fill=#1,opacity=.20,fit={(A)($(B)+(\linewidth,.8\baselineskip)$)}] {};}\ignorespaces}
\newcommand{\norm}[1]{\left\lVert#1\right\rVert}

\title{Byzantine-Robust and Privacy-Preserving\\ Framework for FedML}

\iclrfinalcopy
\author{
{\rm Hanieh Hashemi}\\
ECE Department, USC\\
hashemis@usc.edu
\and
{\rm Yongqin Wang}\\
ECE Department, USC\\
yongqin@usc.edu
\and
{\rm Chuan Guo}\\
Facebook AI Research\\
chuanguo@fb.com
\AND
{\rm Murali Annavaram}\\
ECE Department, USC\\
annavara@usc.edu
}

\begin{document}
\maketitle

\begin{abstract}
Federated learning has emerged as a popular paradigm for collaboratively training a model from data distributed among a set of clients. This learning setting presents, among others, two unique challenges: how to protect privacy of the clients' data during training, and how to ensure integrity of the trained model.
We propose a two-pronged solution that aims to address both challenges under a single framework.
First, we propose to create secure enclaves using a trusted execution environment (TEE) within the server. Each client can then encrypt their gradients and send them to verifiable enclaves. The gradients are decrypted within the enclave without the fear of privacy breaches. However, robustness check computations in a TEE are computationally prohibitive. Hence, in the second step, we perform a novel gradient encoding that enables TEEs to encode the gradients and then offloading Byzantine check computations to accelerators such as GPUs. Our proposed approach provides theoretical bounds on information leakage and offers a significant speed-up over the baseline in empirical evaluation.

\end{abstract}
\section{Introduction}
Deep learning has various applications from health care and smart homes to autonomous vehicles and personal assistants. In these applications, the training data contains highly sensitive personal information such as medical condition and geographical location that the client may be unwilling to share due to privacy concerns. In recent years, federated learning has emerged as a promising solution for facilitating privacy-preserving model training over sensitive data~\citep{mcmahan2017communication}, whereby multiple parties collaboratively train a shared model by sending their gradient updates to a central parameter server without revealing their private data directly to one another. To protect against leakage of information from gradients, secure aggregation mechanisms can be utilized to obfuscate gradients passed to the parameter server so that only the aggregated gradient is revealed~\citep{bonawitz2017practical, bell2020secure, so2020scalable}.
However, these approaches suffer from the lack of integrity checks, as Byzantine clients can inject poisoned gradients to affect the joint model's accuracy and convergence~\citep{blanchard2017machine, bhagoji2019analyzing, bagdasaryan2020backdoor}. These conflicting goals of protecting training data privacy and ensuring model integrity have mostly been tackled separately in prior work.

We present a unified framework for private and Byzantine-robust federated training using a two-pronged strategy. First, we rely on secure hardware enclaves, such as Intel SGX~\citep{costan2016intel}, to process encrypted gradients from clients on the parameter server. 
In the second step, we employ robust gradient aggregation~\citep{blanchard2017machine} to detect Byzantine clients by computing pair-wise distances between gradients from different clients and removing outliers. However, such pair-wise distance computations incur a significant memory and computation cost that scales quadratically with the number of clients, prohibiting its direct use inside the secure enclave. 
To tackle this issue, we propose a mechanism where the CPU-based enclave obfuscates the gradients and offloads pair-wise distance computations to GPUs (or any other accelerator). The novelty of our approach is that the enclave's gradient obfuscation strategy allows GPUs to compute pair-wise distance on encoded gradients. To summarize, our framework simultaneously provides:

\textbf{Security}: Robustness against a subset of Byzantine clients that may send inaccurate gradients to undermine model accuracy and convergence.

\textbf{Information Theoretic Data Privacy}: Clients cannot access other clients' gradients. Untrusted servers can only access encoded gradients. We provide a rigorous analysis to guarantee bounds of information leakage to be infinitesimally small.

\textbf{Practical and Easy to Implement}: Our framework does not rely on slow cryptographic primitives, and it can support floating-point gradients, which enables efficient training and fast implementation.

\textbf{Resiliency and Scalablity}: This protocol is resilient to drop-outs and new users without extra costs.
\section{Problem Statement}
 \begin{figure}
 \begin{center}
 \vspace{-6ex}
\includegraphics[width=0.5\textwidth]{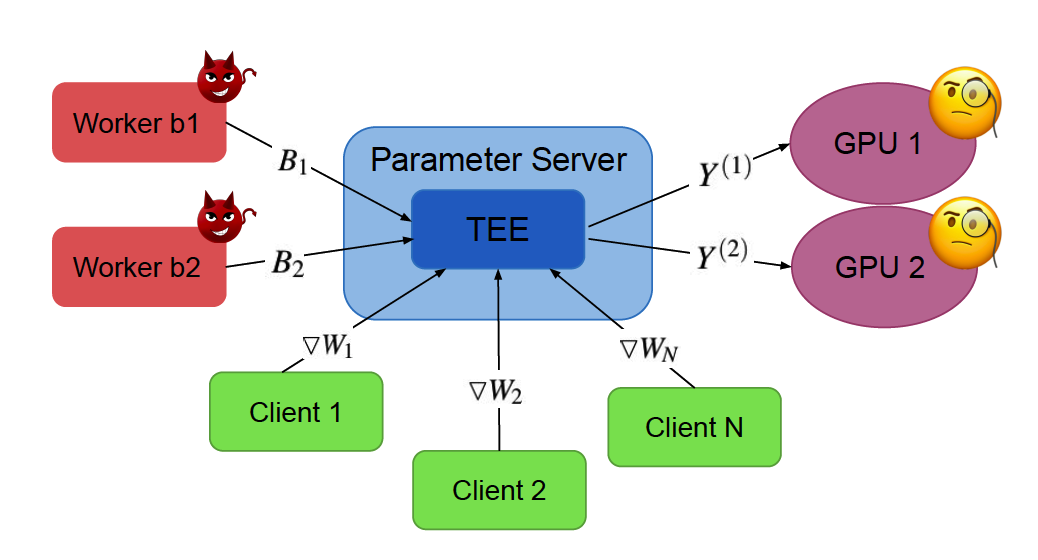}
\caption{Model Structure and Communications.}
 \vspace{-3ex}
\label{fig:overview}
 \end{center}
\end{figure}
\textbf{System Structure}: Our system model is shown in Figure~\ref{fig:overview}. The general structure includes $N$ clients, out of which $f$ are Byzantine. The gradients that these Byzantine clients may send could be potentially disruptive to model convergence and accuracy. These clients communicate with a parameter server for gradient aggregation. The parameter server may in turn enlist additional accelerators like GPUs to perform complex computations such as outlier detection. In Figure~\ref{fig:overview} we show two GPUs that participate in outlier detection ($\text{GPU}_1, \text{GPU}_2$).

\textbf{Threat Model}: The threat model on server-side is classified as a traditional \textit{honest-but-curious (semi-honest)}. This means that the parameter server and the GPUs do not deviate from the agreed-upon protocol, however, they may try to glean private information from what is shared with them. Hence, if these servers receive raw gradients they may use known techniques to extract private information about the client data. To prevent gradient leakage, we assume that the parameter server uses a trusted execution environment (TEE) based enclave. 
We assume that the communication between the TEE parameter server and GPUs is encrypted. 
To make sure that there is a pairwise secure channel between TEE and each GPU, we can use a secret key exchange such as Diffie-Hellman~\citep{steiner1996diffie} at the beginning of the session and encrypt all the messages that leave the TEE using the secret key. The parameter server and GPU accelerators do not collude with each other. On the client side, a fraction of clients can be Byzantine. They can generate erroneous data deliberately in an attempt to sabotage the training.

\textbf{Information Theoretic Data Privacy}: Input privacy is guaranteed by protecting the gradient updates. Note that in our model clients only communicate raw gradients with the TEE-based parameter server. Clients cannot extract any information from other clients' data. 
TEEs will encode data using the novel noise addition scheme described in the next section which is then exposed to GPUs. The amount of leaked information from encoded data to the distance computing GPU servers is rigorously bounded by the noise power regardless of the adversary's computation power. 

\textbf{Robustness}: For robustness, we provide ($\alpha,f$)-Byzantine Resilience that is proposed in~\citep{blanchard2017machine} that can tolerate up to $f$ Byzantine clients if $N\geq 2f+3$.

\section{Algorithm}
\vskip -0.2cm
\label{sec:algori}
 \subsection{Protecting Client Data Privacy}
 \label{sec:privacy1}
Prior works demonstrated that raw gradient can leak information about a client's training data~\citep{geiping2020inverting,zhu2019deep}. Hence, to protect data privacy, untrusted hardware on the server-side should not observe the individual gradients. However, observing these raw gradients is critical for detecting malicious clients~\citep{blanchard2017machine,mhamdi2018hidden,chen2017distributed,yang2019byzantine,yin2018byzantine,peng2020byzantine,pan2020justinian}. To solve this challenge, the parameter server is instantiated as a secure enclave using TEE. 
Each client encrypts its gradients using a mutually agreed-upon key between each client and TEE. Then each client communicates the encrypted gradients with the TEE. The gradients are decrypted within the TEE to get the plain text gradient values. The goal of the TEE is to aggregates the gradients using a robust aggregation function, which detects any outliers, as we explain in Section~\ref{sec:robust}.  
 \subsection{Robust Aggregation}
  \label{sec:robust}
 \begin{figure}[tbp]
  \centering
 \begin{subfigure}[b]{0.45\textwidth}
\includegraphics[width=\textwidth]{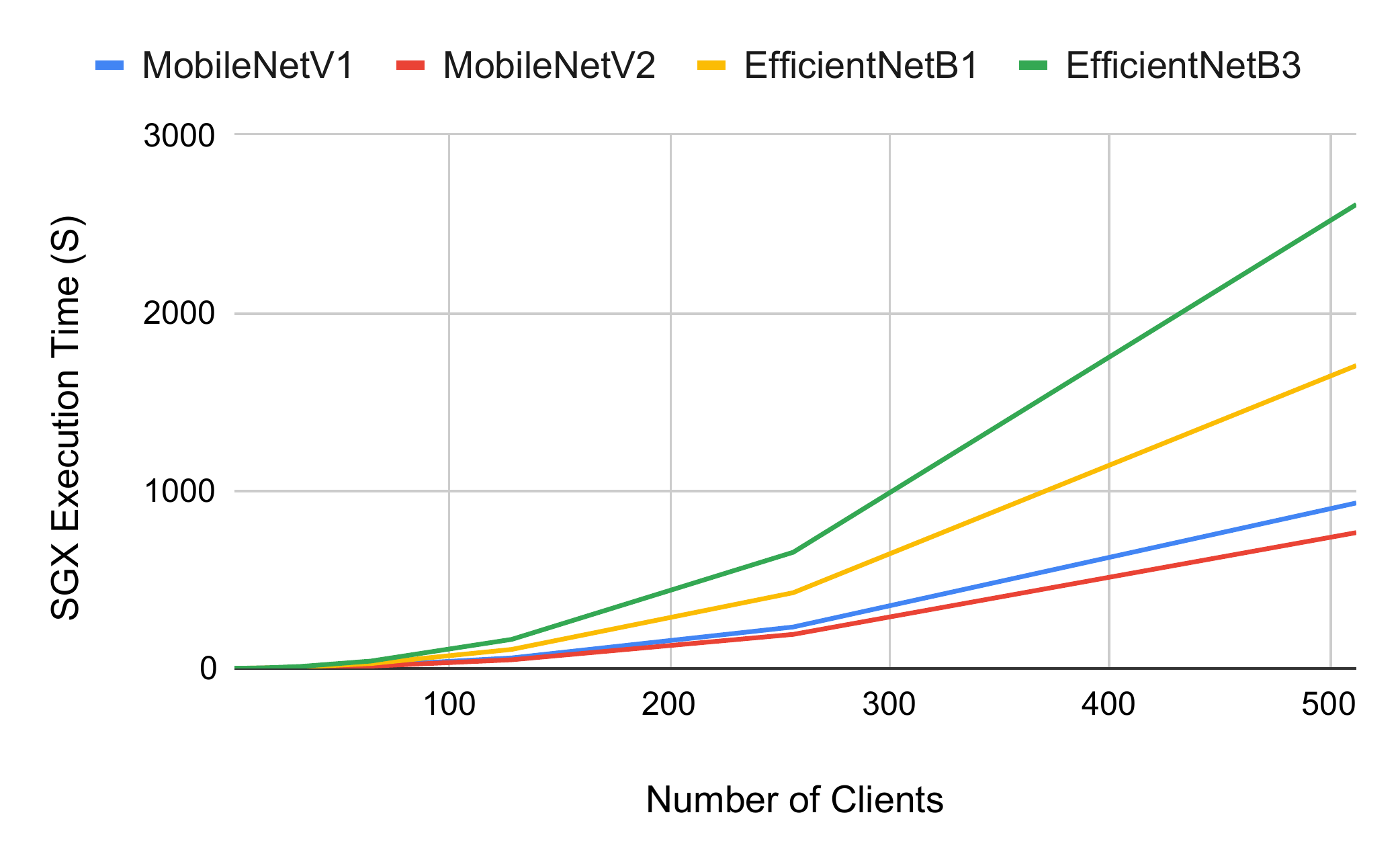}
\label{fig:CPU}
\end{subfigure}
\begin{subfigure}[b]{0.45\textwidth}
\includegraphics[width =\textwidth]{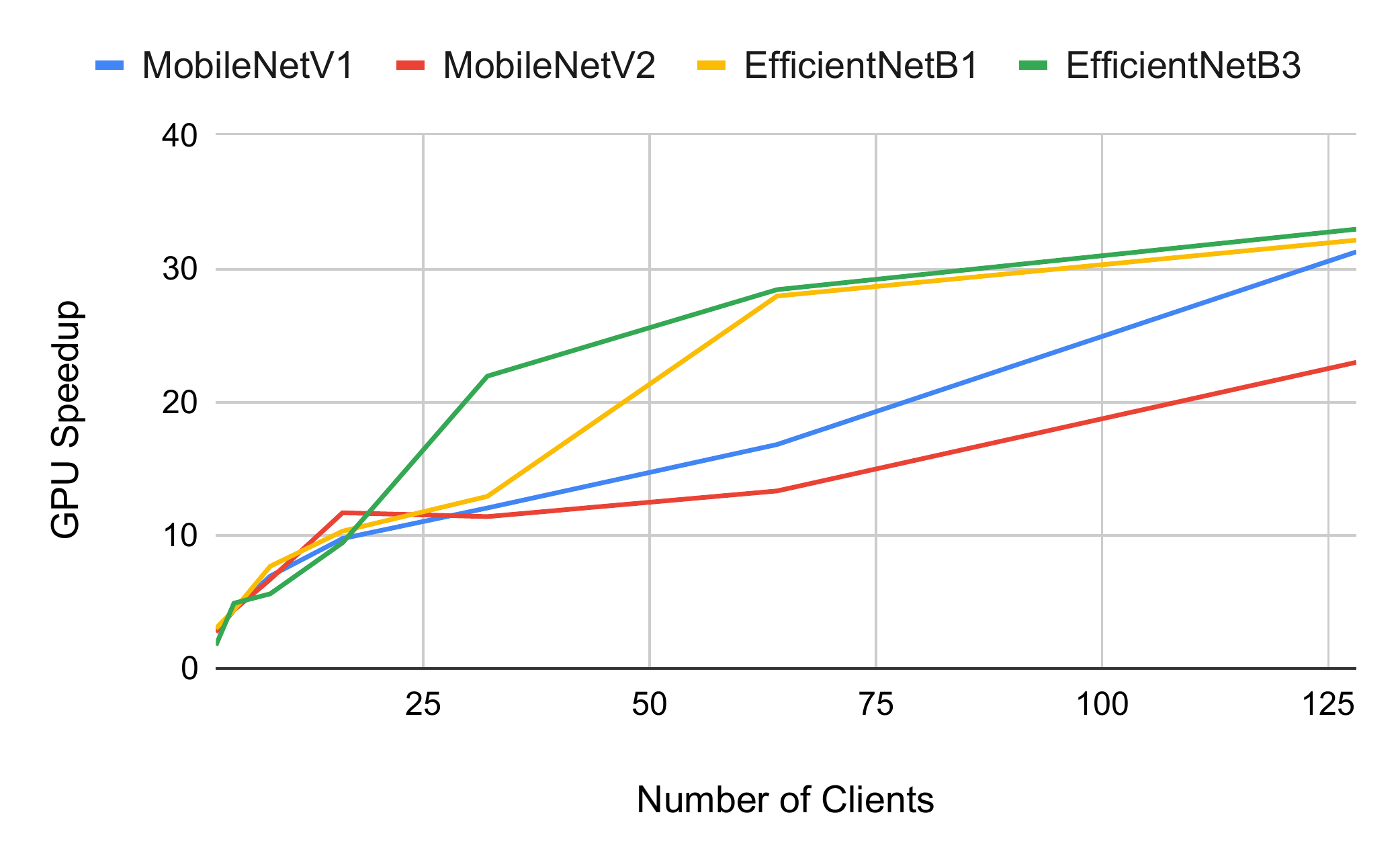}
\label{fig:CPUVSGPU}
\end{subfigure}
\vskip -0.75cm
\caption{(a) SGX pair-wise distance computation time grows quadratically by increasing number of clients for various networks (b) GPU speedup relative to CPU for computing pair-wise distances with different number of clients in various networks.}
\label{fig:CPU2}
\vskip -0.2cm
\end{figure} 
In traditional federated averaging, Byzantine clients can modify the weight update to the desired direction. To prevent such malicious attempts, different robustness algorithms have been proposed. Many robustness schemes, such as  Multi-Krum~\citep{blanchard2017machine} which we use in this work, identify Byzantine clients by computing pair-wise distances on gradients. However, computing pair-wise distances require gradients to be made visible to the parameter server. In our approach, we enable the parameter server to compute pair-wise distances using encoded gradients which bound the information leakage and yet allow GPUs to identify outliers efficiently. 

GPUs are ideally suited for the highly data intensive and vectorized computations of outlier detection. Figure~\ref{fig:CPU2}b shows the speed up for pairwise-distance computations of the same networks on a consumer grade GPU GTX1080Ti relative to TEE. As demonstrated in the figure, when the number of clients grows, GPU shows a significant speedup. Therefore, to overcome TEE's performance bottleneck, a collaboration between GPUs and TEE is proposed. In our approach, TEE protects the data confidentiality by encoding gradients, while Multi-Krum computations are offloaded to the untrusted GPUs. 
The novelty of our encoding is that it is designed in a way that does not affect the accuracy of outlier detection on GPUs. Now, we describe the encoding process in detail.

\textbf{Gradient Encoding}: 
We denote the gradient update that is sent by $\text{client}_i$ as $\triangledown W_i \in \mathbb{R}^d$. In a system with $N$ clients, the goal is to compute the score of the nodes as a measure of how much we trust its gradient update. In order to do so, at the first step, we need to compute the pairwise distances between every two gradient updates $Dist_{i,j} = \norm{\triangledown W_i-\triangledown W_j}^2 \quad \forall i,j \in {1,...,N}$. But for protecting data, we should not reveal the gradients directly to the GPUs. We achieve this goal by adding random noise signals  ($R \in \mathbb{R}^d$) to the clients' gradients.

\textbf{Noise Generation}:
The noise signals are generated at the offline phase within the TEE and one noise signal is generated per each clients' update. As we explain in the next section, the novelty of our outlier detection relies on one restriction placed noise generation. Namely, the $L_2$ norm between any two noise signals generated must be constant ($C$), i.e. $\norm{R_i-R_j}^2 = C$. This restriction allows the TEE to offload distance calculations to GPUs using encoded data while enabling the TEE to decode the distances. Note that for efficiency reasons, the noise can be pre-generated before starting the training process. When there are a large number of clients and insufficient memory within TEE to hold all the random noises, these noises can be encrypted and stored in the external memory that is outside of the TEE. Prior to the outlier detection, the encoding process follows these steps:
      1) Fetch an encrypted random noise vector ($\bar{R_i}\in \mathbb{R}^d$). 
      2) Decrypts the noise vector ($R_i$) within the TEE. 
      3) Compute $Y_{i}^{(1)} = \triangledown W_i+R_i$ and  $Y_{i}^{(2)} = \triangledown W_i-R_i$.  
      4) Encrypt $Y_{i}^{(1)}$ using Diffie Hellmen key of $\text{GPU}_1$ and sends it to the $\text{GPU}_1$. 
      5) Encrypt $Y_{i}^{(2)}$ using Diffie Hellmen key of $\text{GPU}_2$ and sends it to the $\text{GPU}_2$.
 The encoding function repeats the above procedure for all the clients' gradients. 
 
\textbf{Distance Computation on Encoded Data}:
At this step, each GPU has a masked share of each client's gradient and computes the pair-wise distances between them. In other words, for every client $i$ and $j$, $\text{GPU}_1$ computes Equation~\ref{eq:distance} and Similarly, $\text{GPU}_2$ computes Equation~\ref{eq:distance2}:
 \begin{align}\label{eq:distance}
 &\text{Dist}_{i,j}^{(1)} = \norm{Y_{i}^{(1)}-Y_{j}^{(1)}}^2= \norm{\triangledown W_i-\triangledown W_j}^2 + \norm{R_i-R_j}^2 + 2(\triangledown W_i-\triangledown W_j)^T(R_i - R_j)
 \end{align} 
 \begin{align}\label{eq:distance2}
 &\text{Dist}_{i,j}^{(2)} = \norm{Y_{i}^{(2)}-Y_{j}^{(2)}}^2= \norm{\triangledown W_i-\triangledown W_j}^2 + \norm{R_i-R_j}^2 - 2(\triangledown W_i-\triangledown W_j)^T(R_i - R_j)
 \end{align}
\textbf{Decoding Pair-Wise Distances}:
The pair-wise distance computations from GPUs are then returned to the TEE based parameter server. TEE collects and aggregates the partial distances from the GPUs. 
 \begin{align}\label{eq:coll}
\text{Dist}_{i,j} = \text{Dist}_{i,j}^{(1)} +\text{Dist}_{i,j}^{(2)} = 2 \norm{\triangledown W_i-\triangledown W_j}^2 + 2 \norm{R_i-R_j}^2
 \end{align}
 
The above equation has an extra offset of $\norm{R_i-R_j}^2$ in addition to the distance computations. The critical observation here is that outlier detection \textit{does not need} the exact value of distances for selecting outliers. Therefore, if $\norm{R_i-R_j}^2 = C$ where $C$ is a constant for all $i,j$, the relative difference in $L_2$ distance between gradients is preserved and outlier detection works without any limitation. These noise signals can be generated in the offline phase by TEE, encrypted, and stored in the parameter server. Hence, the noise generation overhead is minimized to the noise decryption.
 
\textbf{Gradient Aggregation Rule}: 
After decoding the distances, for each client $i$, its $(N-f-2)$ closest distances are selected (\text{Sort($\text{Dist}[i])\quad \forall i \in {1,..,N}$}). These distances are summed up for these close neighbors to determine that client's score, \text {score($\text{client}_i$) = $\sum_{j=1}^{N-f-2} Dist_{i,j}$}. The lower the score is the closer gradient is to the majority of clients. Top $K$ gradients are selected to update the model. After choosing $K$ clients with the lowest scores, finally the TEE can combine their gradient updates to generate the updated model for the next iteration of model training. 

\textbf{Formal Bounds on Information Leakage}: Our gradient encoding based on noise addition leaks a small amount of information from the client's gradient, which we measure using information-theoretic notions as in prior work~\citep{aliasgari2013secure, guo2020secure}. The exact amount of information leakage in relation to the noise magnitude follows as a corollary of the classic Gaussian channel coding theorem~\citep{cover1999elements}; please see appendix for a derivation.
\begin{cor}
Suppose $\triangledown W_1,\ldots,\triangledown W_N$ are independent gradient vectors. For any $i=1,\ldots,N$, the maximum amount of information each $GPU_j\quad j=1,2$ can retrieve about $\triangledown W_i$ is bounded by:
\begin{align}
    I(\triangledown W_i;Y_1^{(j)},\dots, Y_N^{(j)}) = I(\triangledown W_i;Y_i^{(j)})  \leq \sum_{k=1}^d \frac{1}{2} \log \left(1+\frac{\text{Var}(\triangledown (W_i)_k)}{\sigma^2}\right),
\end{align}
where $Y_i^{(1)}= \triangledown W_i+R_i$,  $Y_i^{(2)}= \triangledown W_i-R_i$ and $R_i\sim \mathcal N(0,\sigma^2 I_d)$.
\end{cor}

\section{Experiments}
\begin{figure}
\begin{center}
\includegraphics[width = 0.40\textwidth]{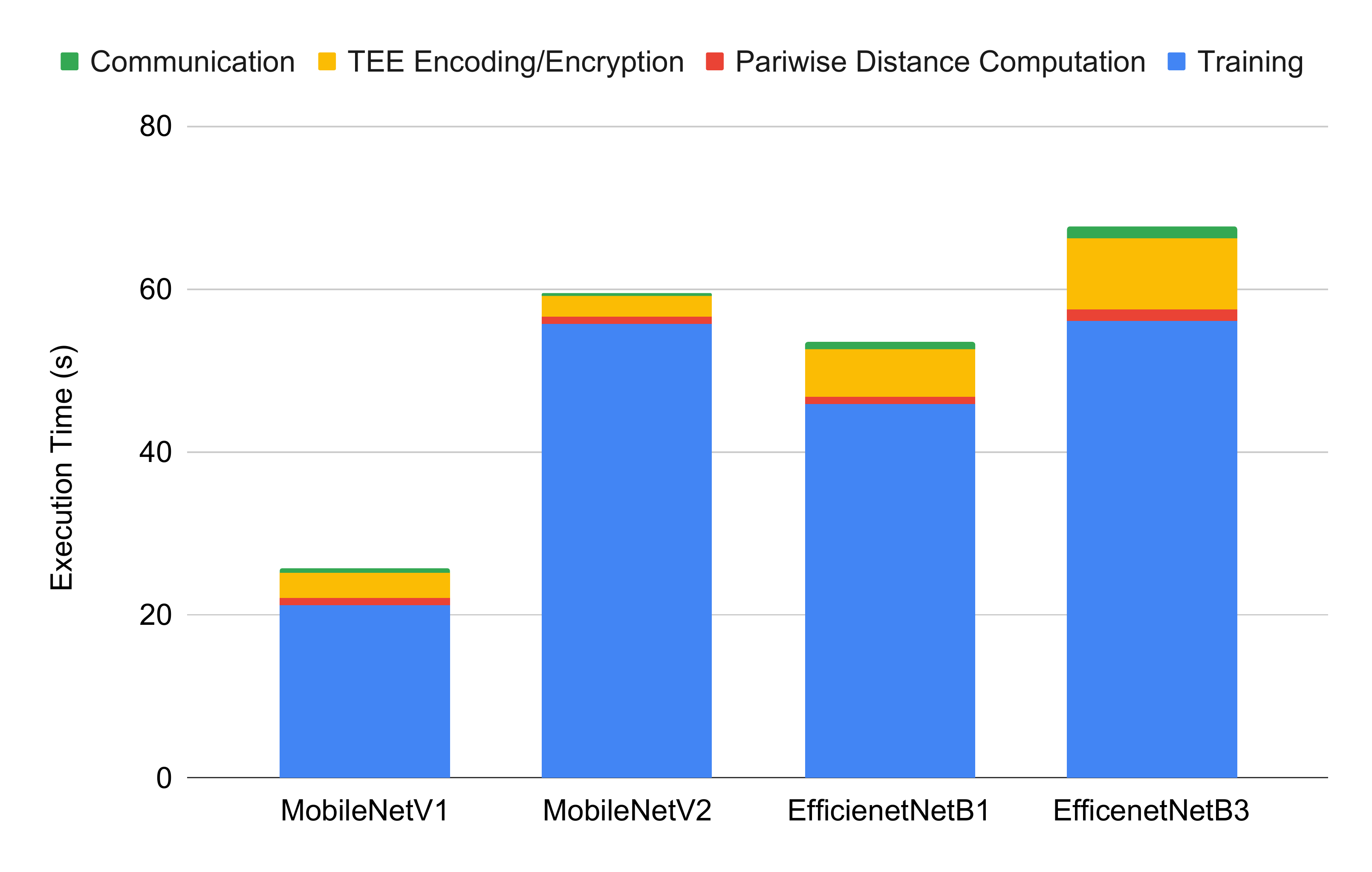}
\caption{Execution time breakdown of CIFAR-10 for different networks.}
\label{fig:breakdown}
\end{center}
\end{figure}
We implemented the proposed design on a GPU-enabled TEE server. Our server consisted of an Intel(R) Coffee Lake E-2174G 3.80GHz processor and two Nvidia GeForce GTX 1080 Ti GPUs. The server has 64 GB RAM and supports Intel Soft Guard Extensions (SGX). For our GPU implementations, we used PyTorch 1.7.1 and Python 3.6.8. To evaluate the performance of our model, we designed a baseline that implements the pair-wise distance computations entirely within Intel SGX, which protects the data privacy but could be prohibitively slow as shown in Figure \ref{fig:CPU2}. We evaluate our method by training two neural network architectures, MobileNet~\citep{sandler2018mobilenetv2,howard2017mobilenets} and EfficientNet~\citep{tan2019efficientnet},  on CIFAR-10, similar to prior works~\citep{li2019survey,malekzadeh2021dopamine}. 

\textbf{Execution Time}:
Figure~\ref{fig:breakdown} depicts the execution time breakdown for one training epoch in a system with 64 clients. For all of the four networks, training time on the edge device dominates the overall execution time. The overhead of the pairwise distance computation on GPU and LAN communication cost using a 1GBps switch is very low, and the main overhead of our scheme comes from the encoding and encryption cost inside the CPU. This cost is measured when using just 1 CPU thread and it can be further improved when using multiple threads in SGX or using pipelining techniques.
\section{Conclusion}
In this paper, we introduced a unified framework for privacy preserving and Byzantine client detection in federated learning. Gradient based information leakage is prevented by sending the raw gradients to a TEE. The code executed within the TEE can be verified by the client. The TEE in turn is responsible for eliminating Byzantine client updates using pair-wise distance computations. But such robustness checks are computationally intensive. Hence, the TEE utilizes GPU hardware for performing this computation. Rather than revealing raw gradients to the GPU our scheme encodes the gradients with random noise, that satisfy a specific criteria. The pair-wise distance computations using gradients encoded with random noise can then be easily decoded to identify outlier clients. Therefore, our scheme relies on the collaboration between TEE, GPU and data encoding process. We designed and implemented our scheme on a realistic server settings and demonstrate more than 6X speedup compared to running the outlier detection entirely within the TEE. 
\newpage

\onecolumn{
\bibliography{iclr2021_conference}

\begin{thebibliography}{26}
\providecommand{\natexlab}[1]{#1}
\providecommand{\url}[1]{\texttt{#1}}
\expandafter\ifx\csname urlstyle\endcsname\relax
  \providecommand{\doi}[1]{doi: #1}\else
  \providecommand{\doi}{doi: \begingroup \urlstyle{rm}\Url}\fi

\bibitem[Aliasgari et~al.(2013)Aliasgari, Blanton, Zhang, and
  Steele]{aliasgari2013secure}
Mehrdad Aliasgari, Marina Blanton, Yihua Zhang, and Aaron Steele.
\newblock Secure computation on floating point numbers.
\newblock In \emph{NDSS}, 2013.

\bibitem[Bagdasaryan et~al.(2020)Bagdasaryan, Veit, Hua, Estrin, and
  Shmatikov]{bagdasaryan2020backdoor}
Eugene Bagdasaryan, Andreas Veit, Yiqing Hua, Deborah Estrin, and Vitaly
  Shmatikov.
\newblock How to backdoor federated learning.
\newblock In \emph{International Conference on Artificial Intelligence and
  Statistics}, pp.\  2938--2948. PMLR, 2020.

\bibitem[Bell et~al.(2020)Bell, Bonawitz, Gasc{\'o}n, Lepoint, and
  Raykova]{bell2020secure}
James~Henry Bell, Kallista~A Bonawitz, Adri{\`a} Gasc{\'o}n, Tancr{\`e}de
  Lepoint, and Mariana Raykova.
\newblock Secure single-server aggregation with (poly) logarithmic overhead.
\newblock In \emph{Proceedings of the 2020 ACM SIGSAC Conference on Computer
  and Communications Security}, pp.\  1253--1269, 2020.

\bibitem[Bhagoji et~al.(2019)Bhagoji, Chakraborty, Mittal, and
  Calo]{bhagoji2019analyzing}
Arjun~Nitin Bhagoji, Supriyo Chakraborty, Prateek Mittal, and Seraphin Calo.
\newblock Analyzing federated learning through an adversarial lens.
\newblock In \emph{International Conference on Machine Learning}, pp.\
  634--643. PMLR, 2019.

\bibitem[Blanchard et~al.(2017)Blanchard, Guerraoui, Stainer,
  et~al.]{blanchard2017machine}
Peva Blanchard, Rachid Guerraoui, Julien Stainer, et~al.
\newblock Machine learning with adversaries: Byzantine tolerant gradient
  descent.
\newblock In \emph{Advances in Neural Information Processing Systems}, pp.\
  119--129, 2017.

\bibitem[Bonawitz et~al.(2017)Bonawitz, Ivanov, Kreuter, Marcedone, McMahan,
  Patel, Ramage, Segal, and Seth]{bonawitz2017practical}
Keith Bonawitz, Vladimir Ivanov, Ben Kreuter, Antonio Marcedone, H~Brendan
  McMahan, Sarvar Patel, Daniel Ramage, Aaron Segal, and Karn Seth.
\newblock Practical secure aggregation for privacy-preserving machine learning.
\newblock In \emph{Proceedings of the 2017 ACM SIGSAC Conference on Computer
  and Communications Security}, pp.\  1175--1191, 2017.

\bibitem[Chen et~al.(2017)Chen, Su, and Xu]{chen2017distributed}
Yudong Chen, Lili Su, and Jiaming Xu.
\newblock Distributed statistical machine learning in adversarial settings:
  Byzantine gradient descent.
\newblock \emph{Proceedings of the ACM on Measurement and Analysis of Computing
  Systems}, 1\penalty0 (2):\penalty0 1--25, 2017.

\bibitem[Costan \& Devadas(2016)Costan and Devadas]{costan2016intel}
Victor Costan and Srinivas Devadas.
\newblock Intel sgx explained.
\newblock \emph{IACR Cryptology ePrint Archive}, 2016\penalty0 (086):\penalty0
  1--118, 2016.

\bibitem[Cover(1999)]{cover1999elements}
Thomas~M Cover.
\newblock \emph{Elements of information theory}.
\newblock John Wiley \& Sons, 1999.

\bibitem[Daniel et~al.(1976)Daniel, Gragg, Kaufman, and
  Stewart]{daniel1976reorthogonalization}
James~W Daniel, Walter~Bill Gragg, Linda Kaufman, and Gilbert~W Stewart.
\newblock Reorthogonalization and stable algorithms for updating the
  gram-schmidt qr factorization.
\newblock \emph{Mathematics of Computation}, 30\penalty0 (136):\penalty0
  772--795, 1976.

\bibitem[Geiping et~al.(2020)Geiping, Bauermeister, Dr{\"o}ge, and
  Moeller]{geiping2020inverting}
Jonas Geiping, Hartmut Bauermeister, Hannah Dr{\"o}ge, and Michael Moeller.
\newblock Inverting gradients--how easy is it to break privacy in federated
  learning?
\newblock \emph{arXiv preprint arXiv:2003.14053}, 2020.

\bibitem[Guo et~al.(2020)Guo, Hannun, Knott, van~der Maaten, Tygert, and
  Zhu]{guo2020secure}
Chuan Guo, Awni Hannun, Brian Knott, Laurens van~der Maaten, Mark Tygert, and
  Ruiyu Zhu.
\newblock Secure multiparty computations in floating-point arithmetic.
\newblock \emph{arXiv preprint arXiv:2001.03192}, 2020.

\bibitem[Howard et~al.(2017)Howard, Zhu, Chen, Kalenichenko, Wang, Weyand,
  Andreetto, and Adam]{howard2017mobilenets}
Andrew~G Howard, Menglong Zhu, Bo~Chen, Dmitry Kalenichenko, Weijun Wang,
  Tobias Weyand, Marco Andreetto, and Hartwig Adam.
\newblock Mobilenets: Efficient convolutional neural networks for mobile vision
  applications.
\newblock \emph{arXiv preprint arXiv:1704.04861}, 2017.

\bibitem[Li et~al.(2019)Li, Wen, Wu, Hu, Wang, and He]{li2019survey}
Qinbin Li, Zeyi Wen, Zhaomin Wu, Sixu Hu, Naibo Wang, and Bingsheng He.
\newblock A survey on federated learning systems: vision, hype and reality for
  data privacy and protection.
\newblock \emph{arXiv preprint arXiv:1907.09693}, 2019.

\bibitem[Malekzadeh et~al.(2021)Malekzadeh, Hasircioglu, Mital, Katarya,
  Ozfatura, and G{\"u}nd{\"u}z]{malekzadeh2021dopamine}
Mohammad Malekzadeh, Burak Hasircioglu, Nitish Mital, Kunal Katarya,
  Mehmet~Emre Ozfatura, and Deniz G{\"u}nd{\"u}z.
\newblock Dopamine: Differentially private federated learning on medical data.
\newblock \emph{arXiv e-prints}, pp.\  arXiv--2101, 2021.

\bibitem[McMahan et~al.(2017)McMahan, Moore, Ramage, Hampson, and
  y~Arcas]{mcmahan2017communication}
Brendan McMahan, Eider Moore, Daniel Ramage, Seth Hampson, and Blaise~Aguera
  y~Arcas.
\newblock Communication-efficient learning of deep networks from decentralized
  data.
\newblock In \emph{Artificial Intelligence and Statistics}, pp.\  1273--1282.
  PMLR, 2017.

\bibitem[Mhamdi et~al.(2018)Mhamdi, Guerraoui, and Rouault]{mhamdi2018hidden}
El~Mahdi~El Mhamdi, Rachid Guerraoui, and S{\'e}bastien Rouault.
\newblock The hidden vulnerability of distributed learning in byzantium.
\newblock \emph{arXiv preprint arXiv:1802.07927}, 2018.

\bibitem[Pan et~al.(2020)Pan, Zhang, Wu, Xiao, Ji, and Yang]{pan2020justinian}
Xudong Pan, Mi~Zhang, Duocai Wu, Qifan Xiao, Shouling Ji, and Zhemin Yang.
\newblock Justinian's gaavernor: Robust distributed learning with gradient
  aggregation agent.
\newblock In \emph{29th $\{$USENIX$\}$ Security Symposium ($\{$USENIX$\}$
  Security 20)}, pp.\  1641--1658, 2020.

\bibitem[Peng \& Ling(2020)Peng and Ling]{peng2020byzantine}
Jie Peng and Qing Ling.
\newblock Byzantine-robust decentralized stochastic optimization.
\newblock In \emph{ICASSP 2020-2020 IEEE International Conference on Acoustics,
  Speech and Signal Processing (ICASSP)}, pp.\  5935--5939. IEEE, 2020.

\bibitem[Sandler et~al.(2018)Sandler, Howard, Zhu, Zhmoginov, and
  Chen]{sandler2018mobilenetv2}
Mark Sandler, Andrew Howard, Menglong Zhu, Andrey Zhmoginov, and Liang-Chieh
  Chen.
\newblock Mobilenetv2: Inverted residuals and linear bottlenecks.
\newblock In \emph{Proceedings of the IEEE conference on computer vision and
  pattern recognition}, pp.\  4510--4520, 2018.

\bibitem[So et~al.(2020)So, Guler, and Avestimehr]{so2020scalable}
Jinhyun So, Basak Guler, and Salman Avestimehr.
\newblock A scalable approach for privacy-preserving collaborative machine
  learning.
\newblock \emph{Advances in Neural Information Processing Systems}, 33, 2020.

\bibitem[Steiner et~al.(1996)Steiner, Tsudik, and Waidner]{steiner1996diffie}
Michael Steiner, Gene Tsudik, and Michael Waidner.
\newblock Diffie-hellman key distribution extended to group communication.
\newblock In \emph{Proceedings of the 3rd ACM conference on Computer and
  communications security}, pp.\  31--37, 1996.

\bibitem[Tan \& Le(2019)Tan and Le]{tan2019efficientnet}
Mingxing Tan and Quoc~V Le.
\newblock Efficientnet: Rethinking model scaling for convolutional neural
  networks.
\newblock \emph{arXiv preprint arXiv:1905.11946}, 2019.

\bibitem[Yang et~al.(2019)Yang, Zhang, Fang, and Liu]{yang2019byzantine}
Haibo Yang, Xin Zhang, Minghong Fang, and Jia Liu.
\newblock Byzantine-resilient stochastic gradient descent for distributed
  learning: A lipschitz-inspired coordinate-wise median approach.
\newblock In \emph{2019 IEEE 58th Conference on Decision and Control (CDC)},
  pp.\  5832--5837. IEEE, 2019.

\bibitem[Yin et~al.(2018)Yin, Chen, Ramchandran, and
  Bartlett]{yin2018byzantine}
Dong Yin, Yudong Chen, Kannan Ramchandran, and Peter Bartlett.
\newblock Byzantine-robust distributed learning: Towards optimal statistical
  rates.
\newblock \emph{arXiv preprint arXiv:1803.01498}, 2018.

\bibitem[Zhu et~al.(2019)Zhu, Liu, and Han]{zhu2019deep}
Ligeng Zhu, Zhijian Liu, and Song Han.
\newblock Deep leakage from gradients.
\newblock In \emph{Advances in Neural Information Processing Systems}, pp.\
  14774--14784, 2019.

\end{thebibliography}
\bibliographystyle{iclr2021_conference}
}

\setcounter{cor}{0}
\newpage
\appendix
\section{Appendix}
\subsection{Information Theoretic Data Privacy}
\label{sec:privacy}
In this section, we measure the amount of information the adversary can potentially gain about the raw gradients from the encoded gradients even if the adversary has access to an unlimited computation power. Information-theoretic privacy is the strongest privacy measurement since it does not assume any limitation on the adversary's computational power. The amount of information leaked by $\triangledown W_i+R_i$'s about $\triangledown W_i$ is the \textbf{mutual information (MI)} between these two sets of variables, defined by~\citep{cover1999elements}:
\begin{align}
    I(\triangledown W_i; \triangledown W_i+R_i)=H(\triangledown W_i)-H(\triangledown W_i |\triangledown W_i+R_i)~.
\end{align}
Here, $ H(\cdot)$ denotes the Shannon entropy. Note that the information that adversary can potentially learn about $\triangledown W_i$ by having $\triangledown W_i+R_i$ is fundamentally bounded by  $I(\triangledown W_i ; \triangledown W_i+R_i)$, which is a well-known information-theory concept called parallel Gaussian channel capacity~\citep{cover1999elements}.\\

\begin{thm}[Section 9.4 of~\cite{cover1999elements}]\label{thm:bound_sum}
Assume that each $X_i\sim P_{X_i}$ is a random variable with bounded variance $\text{Var}(X_i)$, and $R_i\sim \mathcal N(0,\sigma_i^2)$ are independent Gaussian random variables with variance $\sigma_i^2$ and mean $0$. Then:
\begin{align}
    I(X_1, X_2,..., X_n;X_1+R_1,X_2+R_2,...,X_n+R_n) \leq\sum_{i=1}^n \frac{1}{2} \log \left(1+\frac{\text{Var}(X_i)}{\sigma^2_{i}}\right).
\end{align}
\end{thm}

We utilize this theorem in the following corollary to bound the information leakage in our framework. Note that since the two GPUs are independent and non-colluding, their information gain is independent of each other.

\begin{cor}
Suppose $\triangledown W_1,\ldots,\triangledown W_N$ are independent gradient vectors. For any $i=1,\ldots,N$, the maximum amount of information each $GPU_j\quad j=1,2$ can retrieve about $\triangledown W_i$ is bounded by:
\begin{align}
    I(\triangledown W_i;Y_1^{(j)},\dots, Y_N^{(j)}) = I(\triangledown W_i;Y_i^{(j)})  \leq \sum_{k=1}^d \frac{1}{2} \log \left(1+\frac{\text{Var}(\triangledown (W_i)_k)}{\sigma^2}\right),
\end{align}
where $Y_i^{(1)}= \triangledown W_i+R_i$,  $Y_i^{(2)}= \triangledown W_i-R_i$ and $R_i\sim \mathcal N(0,\sigma^2 I_d)$.
\end{cor}

Hence if the gradient vector has a bounded variance\footnote{This can be achieved by clipping the coordinates to a bounded range.}, the mutual information bound can be made arbitrarily small by choosing a noise variance that is orders of magnitude larger than that of the gradient vector.

\subsection{Noise Generation}

Our goal is to generate $N$ random vectors, such that $\norm {\hat R_i-\hat R_j}^2 = C$ for all $i,j$. We generate such random vectors as follows,

\begin{enumerate}
    \item Generate $N$ i.i.d zero-mean Gaussian vectors ($R_i,...,R_N$), where the variance of each entry is $\sigma^2$.
    \item Orthogonalize them using the Gram-Schmidt process~\citep{daniel1976reorthogonalization}. This gives us
    $\bar{R}_1,...,\bar{R}_N$ such that $\forall i,j \quad \bar{R}_i^T\bar{R}_j=0 $
    \item Scale all the vectors so they have the same norm. $\hat{R}_1=\frac{\sqrt C\cdot \bar{R}_1}{\sqrt 2 \norm{\bar{R}_1}},..., \hat{R}_N=\frac{\sqrt C\cdot \bar{R}_N}{\sqrt 2\norm{\bar{R}_N}}$ such that $\norm{\hat{R}_i}=\sqrt{C/2}$ for all $i$.
\end{enumerate}

Now after this procedure, we have
\begin{align}
    \forall i,j: \norm{\hat{R_i}-\hat{R_j}}^2 = \norm{\hat{R_i}}^2+\norm{\hat{R_j}}^2-2\hat{R_i}^T\hat{R_j}=C
\end{align}
Note that this procedure is equivalent to choosing the first $N$ rows of a random orthogonal matrix, and scaling the rows to have the desired norms. Therefore, each $\hat R_i$ has a uniformly random direction in the $d$ dimensional space, same as a Gaussian vector. The only difference between $\hat R_i$ and the initial Gaussian vectors $R_i$ is that its norm has been fixed to be $C$. However, as $d$ grows, the i.i.d. Gaussian vectors $(R_1,\ldots,R_N)$ converge to orthonormal vectors with a fixed norm $\sigma \sqrt{d}$, so the noise generation process roughly reduces to sampling $N$ i.i.d. Gaussian vectors with $\sigma = \sqrt{C / 2d}$.

\end{document}